\documentclass[preprint,12pt]{elsarticle}

\usepackage{graphicx}%
\usepackage{multirow}%
\usepackage{amsmath,amssymb,amsfonts}%
\usepackage{amsthm}%
\usepackage{mathrsfs}%
\usepackage[title]{appendix}%
\usepackage{xcolor}%
\usepackage{textcomp}%
\usepackage{manyfoot}%
\usepackage{booktabs}%
\usepackage{listings}%
\usepackage{url}

\usepackage{tabularray}
\usepackage{overpic}
\usepackage[normalem]{ulem}

\newtheorem{theorem}{Theorem}
\newtheorem{corollary}[theorem]{Corollary}%
\newtheorem{lemma}[theorem]{Lemma}%

\usepackage{mathtools}
\DeclarePairedDelimiter\abs{\lvert}{\rvert}

\DeclarePairedDelimiter\ceil{\lceil}{\rceil}
\DeclarePairedDelimiter\floor{\lfloor}{\rfloor}
\DeclarePairedDelimiter\parenv{\lparen}{\rparen}

\DeclarePairedDelimiter\set{\{}{\}}
\DeclarePairedDelimiter\ang{\langle}{\rangle}
\DeclarePairedDelimiter\mset{\{\!\!\{}{\}\!\!\}}

\renewcommand{\leq}{\leqslant}

\renewcommand{\geq}{\geqslant}

\newcommand{\cB}{\mathcal{B}}
\newcommand{\cC}{\mathcal{C}}

\newcommand{\cP}{\mathcal{P}}

\newcommand{\cV}{\mathcal{V}}

\newcommand{\F}{\mathbb{F}}

\newcommand{\E}{\mathbb{E}}

\newcommand{\eqdef}{\triangleq}

\DeclareMathOperator{\supp}{supp}
\DeclareMathOperator{\wt}{wt}

\DeclareMathOperator{\rank}{rank}
\DeclareMathOperator{\im}{im}
\DeclareMathOperator{\tr}{tr}
\DeclareMathOperator{\PG}{PG}
\DeclareMathOperator{\Bin}{Bin}
\DeclareMathOperator{\rowsp}{rowsp}

\begin{document}

\begin{frontmatter}
\title{On the Generalized Packing and Covering Radii of Codes}

\author[1]{Wenjun Yu} 
\ead{yuwenjun@mail.ustc.edu.cn} 

\author[2]{Moshe Schwartz} 

\ead{schwartz.moshe@mcmaster.ca} 

\address[1]{Institute of Mathematics and Interdisciplinary Sciences, Xidian University, Xi'an, 710126, China} 

\address[2]{Department of Electrical and Computer Engineering, McMaster University, Hamilton, ON L8S 4K1, Canada} 

\begin{abstract}
The minimum distance and the covering radius are two fundamental properties of the code. Both have been extended: the former to the generalized Hamming weights hierarchy, and the latter to the generalized covering radii hierarchy. In both cases, the lowest level of the hierarchies corresponds to the classical minimum distance and covering radius, respectively. From a geometric point of view, the minimum distance of the code determines the packing radius, which is upper bounded by the covering radius. It was conjectured this relation extends to all other orders of the hierarchy, namely, that the generalized packing radii are upper bounded by the generalized covering radii of the same order.

In this paper we prove this conjecture is true for the second order radii. We also prove the conjecture holds for all orders when the code rate is at most $3/5$. Finally, we show that for any code rate in $(0,1)$, for all sufficiently long codes the conjecture holds for all orders.
\end{abstract}


\end{frontmatter}

\section{Introduction}

For a prime power $q$, let $\F_q$ denote the finite field of size
$q$. We use $\F_q^n$ to denote the set of vectors of length $n$ over
$\F_q$, and $\F_q^{t\times n}$ to denote the set of matrices of size
$t\times n$ over $\F_q$. For a matrix $A\in\F_q^{t\times n}$, we use
$\rowsp(A)$ to denote the vector space spanned by the rows of $A$. We
also use $A_{|j}$ to denote the $j$-th column of $A$.

Given a vector, $v=(v_1,\dots,v_n)\in\F_q^n$, the support of $v$ is
defined as
\[ \supp(v) \eqdef \set*{ 1\leq j\leq n : v_j\neq 0}.\]
With this, we can define the Hamming weight of a vector as
$\wt(v)\eqdef \abs{\supp(v)}$, and the Hamming distance between
$v,v'\in\F_q^n$ as $d(v,v')\eqdef \wt(v-v')$.
This is extended to matrices. Let $V\in\F_q^{t\times n}$ be some
matrix. Then
\[ \supp(V) \eqdef \set*{ 1\leq j\leq n : V_{|j}\neq 0}.\]
Naturally, $\wt(V)\eqdef \abs{\supp(V)}$, and for
$V,V'\in\F_q^{t\times n}$, we define $d(V,V')=\wt(V-V')$. This
distance measure gives rise to the block metric~\cite{FenXuHic06}.

An $[n,k]_q$ linear code $\cC\subseteq\F_q^n$ is a vector space, whose
elements are called codewords. We say $n$ is the length of the code,
$k=\dim \cC$ the dimension of the code, and $n-k$ is called the
redundancy of the code. The minimum distance of the code, denoted
$d(\cC)$, is defined as the smallest Hamming distance between two
distinct codewords, i.e.,
\begin{equation}
  \label{eq:defd}
  d(\cC) \eqdef \min_{\substack{c,c'\in\cC \\ c\neq c'}} d(c,c')=\min_{\substack{c\in\cC\\ c\neq 0}} \wt(c).
\end{equation}
If we define
\[ \delta(\cC) \eqdef \floor*{\frac{d(\cC)-1}{2}},\]
then it is known~\cite{MacSlo78} that $\delta(\cC)$ is the packing
radius of $\cC$, namely, the largest radius of balls in Hamming
metric, centered at the codewords, such that any two distinct balls
are disjoint.

The counterpart of the packing radius is the covering radius of the
code, denoted by $R(\cC)$. Geometrically, it is defined as the
smallest radius of balls, centered at the codewords, such that their
union covers the entire space $\F_q^n$. Thus,
\begin{equation}
  \label{eq:defR}
R(\cC) \eqdef \max_{v\in\F_q^n} \min_{c\in\cC}d(v,c).
\end{equation}

Both the minimum Hamming distance of the code, and the covering radius
of the code, have been generalized to hierarchy of parameters, called
the generalized Hamming weights of the code, and the generalized
covering radii of the code, respectively. To present those, we define
for all integers $t\geq 1$,
\[ \cC^t \eqdef \set*{ C\in\F_q^{t\times n} : \text{each row of $C$ is a codeword of $\cC$}}.\]

The $t$-th order generalized hamming weight of $\cC$, denoted
$d_t(\cC)$, was introduced by~\cite{Wei91}, and in our notation is defined as
\begin{equation}
  \label{eq:defdt}
d_t(\cC) \eqdef \min_{\substack{C\in\cC^t \\ \rank C=t}}\wt(C).
\end{equation}
Comparing this with~\eqref{eq:defd}, we see that $d_1(\cC)=d(\cC)$,
and thus the generalized Hamming weight hierarchy of the code,
$d_1(\cC),d_2(\cC),\dots$, generalizes the notion of the minimum
distance of the code.

Much more recently, \cite{EliFirSch21a} introduced the notion of the generalized covering radius. The $t$-th order generalized covering radius of $\cC$,
denoted $R_t(\cC)$, is defined as
\[
R_t(\cC) \eqdef \max_{V\in\F_q^{t\times n}} \min_{C\in\cC^t} d(V,C).
\]
Once again, comparing this with~\eqref{eq:defR}, we see that
$R_1(\cC)=R(\cC)$. Thus, the generalized covering radius hierarchy,
$R_1(\cC), R_2(\cC),\dots$, generalizes the notion of the covering
radius of the code.

In~\cite{EliFirSch21a}, the following quantity
\[
\delta_t(\cC) \eqdef \floor*{\frac{d_t(\cC)-1}{2}},
\]
was dubbed the $t$-order generalized packing radius of the code. It
was conjectured that $\delta_t(\cC)\leq R_t(\cC)$ in the relevant range, i.e., all $1\leq t\leq \min\set{k,n-k}$. This is obviously
true for the classical case of $t=1$ (e.g., see~\cite{CohHonLitLob97}). It is not known
whether this holds for $t\geq 2$.

The generalized Hamming weights were introduced in~\cite{Wei91} 35 years ago. Since then, hundreds of papers followed, studying the properties of this hierarchy of weights, finding the values of some known codes, bounding it for others, and constructing codes with desired hierarchies. In comparison, the more recently introduced generalized covering radius~\cite{EliFirSch21a} has much fewer known results. The generalized covering radii of Hamming, extended Hamming, and repetition codes was studied in~\cite{EliFirSch21a},  BCH codes in~\cite{YohSch25,OzbOzt26,OzbIlk26a,EssZab26,BelZab26,XioYip26,XioYipZul26}, Melas codes~\cite{LiXio26}, other cyclic codes~\cite{LuoZhoMesSagYan26}, and Reed-Muller codes~\cite{EliWeiSch22}. The asymptotic capacity of codes with a specific generalized covering radius was studied in~\cite{EliFirSch21a,EliSch24,LiShaWei26}. We also note~\cite{AlfMarNerTro26}, which studies the generalized covering radii from a geometric perspective while still noting that the conjecture of whether $\delta_t(\cC)\leq R_t(\cC)$ is still open. If the conjecture were to be resolved, then apart from a greater understanding of the geometry of codes, we would automatically be able to use all the extensive results on generalized Hamming weights to lower bound the generalized covering radii of codes, of which we know far less.

The main contributions of this paper are as follows. We prove that the conjecture is true for order $t=2$ for all codes. We then show the conjecture holds for all codes with rate at most $3/5$ for all orders simultaneously. Additionally, we prove an asymptotic result, showing that for any rate $0< \alpha < 1$, there exists $N_\alpha$ (that depends on the rate $\alpha$ only), such that all codes of rate $\alpha$ and length $n\geq N_\alpha$ satisfy the conjecture for all orders simultaneously. Finally, we show two more bounds on the generalized covering radius in terms of the generalized Hamming weight.

The paper is organized as follows. In Section~\ref{sec:prelim} we bring some relevant notation and known results that will be used later. In Section~\ref{sec:t2} we prove the conjecture for $t=2$. In Section~\ref{sec:rate} we prove the conjecture for codes of rate at most $3/5$, and then in Section~\ref{sec:asymptotic} we prove the asymptotic result. Finally, we prove some bounds in Section~\ref{sec:genbound}.

\section{Preliminaries}
\label{sec:prelim}

From a geometric perspective, the covering problem requires the notion of a ball. The Hamming ball of radius $r$ centered at $v\in\F_q^n$ is defined as
\[
\cB_{q,r}(v) \eqdef \set*{ v'\in\F_q^n : d(v',v) \leq r}.
\]
This notion may be extended to the block metric. Here, for $V\in\F_q^{t\times n}$ as the center of the ball,
\[
\cB_{q,r}^{(t)}(V) \eqdef \set*{ V'\in\F_q^{t\times n} : d(V',V)\leq r}.
\]
It was shown in~\cite{EliFirSch21a} that
the volume of a ball of radius $r$ in the order-$t$ block metric of
length $n$, does not depend on the choice of center, and is given by
\[
\cV_{q^t}(n,r) \eqdef \abs*{\cB_{q,r}^{(t)}(V)} = \sum_{i=0}^r \binom{n}{i}(q^t-1)^i.
\]
Additionally, it was shown there that the ball-covering argument
implies that for an $[n,k]_q$ code $\cC$, one must have
\begin{equation}
\label{eq:downball}
\cV_{q^t}(n,R_t(\cC)) \geq q^{t(n-k)}.
\end{equation}
Alternatively,
\begin{equation}
  \label{eq:nec}
  \cV_{q^t}(n,r) < q^{t(n-k)} \implies R_t(\cC) > r.
\end{equation}

The following known results about the generalized Hamming weights will be used in this paper. First, we recall the generalized Singleton bound from~\cite{Wei91}.

\begin{lemma}[{\cite[Theorem 1]{Wei91}}]
\label{lem:knownsingleton}
Let $\cC$ be and $[n,k]_q$ code. Then for all $1\leq t\leq k$,
\[
d_t(\cC) \leq n-k+t.
\]
\end{lemma}

Another lemma is due to~\cite{TsfVla95}:
\begin{lemma}[{\cite[Corollary 3.7]{TsfVla95}}]
\label{lem:knownvla}
Let $\cC$ be an $[n,k]_q$ code. Then for any $0\leq \ell\leq s-r$, $1\leq r\leq s\leq k$,
\[
d_r(\cC) \leq \floor*{\frac{(q^r-1)q^{s-\ell-r}}{q^{s-\ell}-1}(d_s(\cC)-\ell)}.
\]
\end{lemma}

Finally, some miscellaneous notation. A multiset will be denoted by
double curly braces, $\mset{\dots}$. If $A$ is a multiset, we use
$\#_x(A)$ to denote the number of times $x$ appears in $A$. This might
be $0$ if $x\notin A$.

Two vectors $v,v'\in\F_q^{\ell}\setminus\set{0}$ are said to be
equivalent if $v=\alpha v'$ for some
$\alpha\in\F_q\setminus\set{0}$. The equivalence class of $v$ is
denoted by $\ang{v}$ and is called a projective point. The set of
equivalence classes is called a projective space, and denoted
$\PG(\ell-1,q)$.

For $q\geq 2$, the $q$-ary entropy function, $H_q(x):[0,1-q^{-1}]\to [0,1]$ is defined as
\[
H_q(x) \eqdef x\log_q (q-1) - x\log_q x - (1-x)\log_q (1-x),
\]
with the convention that $0\log_q 0=0$. In this range, the function is a bijection, and so the inverse function $H_q^{-1}(x): [0,1]\to [0,1-q^{-1}]$ is well defined.

\section{Order $t=2$}
\label{sec:t2}

In this section we prove the conjecture holds for order $t=2$ for all codes. The proof technique is based on inspecting the multiset of columns of a generator matrix for the code as projective points. By relating their multiplicities to $d_2(\cC)$ and $R_2(\cC)$, the bound is proved.

Before stating and proving this claim, we require two lemmas. The first is a simple technical lemma concerning the partitioning of a multiset.

\begin{lemma}
  \label{lem:partition}
  Let $S$ be a set, and let $T\subseteq S$ be a finite multiset
  containing elements from $S$, with $\abs{T}=d$. Then there exists a
  partition of $T$ into two multisets, $A$ and $B$, $A\cup B=T$, and
  $\abs{A}+\abs{B}=\abs{T}$, such that:
  \begin{itemize}
  \item
    $\abs{A} =\floor{d/2}$, $\abs{B}=\ceil{d/2}$, and
  \item
    for any $x,y\in S$, $x\neq y$, we have $\#_x(A)+\#_y(B) \leq \floor{d/2}+1$.
  \end{itemize}
\end{lemma}

\begin{proof}
  Let $O\subseteq S$ be the set of elements that appear an odd number of
  times in $T$, i.e.,
  \[ O \eqdef \set*{ x\in S : \#_x(T) \text{ is odd} }.\]
  Since $d=\sum_{x\in S}\#_x(T)$, we have $d\equiv \abs{O}\pmod{2}$. Moreover,
  \[
  \sum_{x\in S} \floor*{\frac{\#_x(T)}{2}} = \frac{d-\abs{O}}{2}.
  \]
  We create a partition of $T$ into multisets $A$ and $B$ by
  specifying $A$, and thus, the remaining items are placed in $B$.

  For each $x\in S$, we place $\floor{\#_x(T)/2}$ copies of $x$ in
  $A$. Then, we arbitrarily choose $\floor{\abs{O}/2}$ elements from
  $O$, which are added to $A$ (as a multiset). Thus, each element
  $x\in S$ appears $\floor{\#_x(T)/2}$ or $\ceil{\#_x(T)/2}$ in $A$,
  with the latter option occurring exactly $\floor{\abs{O}/2}$
  times. It follows that
  \[
  \abs*{A} = \sum_{x\in S} \floor*{\frac{\#_x(T)}{2}} + \floor*{\frac{\abs{O}}{2}} = \frac{d-\abs{O}}{2}+\floor*{\frac{\abs{O}}{2}} = \floor*{\frac{d}{2}}.
  \]

  Finally, let $x,y\in S$, $x\neq y$, be two distinct elements. We have
  \begin{align*}
    \#_x(A) &\leq \ceil*{\frac{\#_x(T)}{2}}, \\
    \#_y(B) & = \#_y(T)-\#_y(A) \leq \#_y(T) - \floor*{\frac{\#_y(T)}{2}} = \ceil*{\frac{\#_y(T)}{2}}.
  \end{align*}
  Hence,
  \[
  \#_x(A)+\#_y(B) \leq \ceil*{\frac{\#_x(T)}{2}} + \ceil*{\frac{\#_y(T)}{2}}
  \leq \floor*{\frac{\#_x(T)+\#_y(T)}{2}}+1 \leq \floor*{\frac{d}{2}}+1.
  \]
\end{proof}

In the next lemma, we study $2\times n$ matrices of a form that will later be useful in the main theorem.

\begin{lemma}
  \label{lem:matrix}
  Let $C,Y,Z\in\F_q^{2\times n}$, $n\geq 2$, be such that
  \begin{align*}
    C&=Y+Z, &
    \rank C &= 2, &
    \rank Y&=\rank Z = 1.
  \end{align*}
  Then there exists $P\in\F_q^{2\times 2}$ such that
  \begin{align*}
    Y &= PC, &
    Z &= (I_2-P)C, &
    P^2 &= P, &
    \rank P &= 1.
  \end{align*}
  In particular, $\im P$ and $\ker P$ are distinct one-dimensional
  subspaces of $\F_q^2$ and $\ker(I_2-P)=\im P$.
\end{lemma}

\begin{proof}
  We have
  \[
  \rowsp(C) \subseteq \rowsp(Y)+\rowsp(Z).
  \]
  The LHS has dimension $2$, while that on the right at most $2$. Hence,
  \[
  \rowsp(C) = \rowsp(Y)+\rowsp(Z).
  \]
  In particular, $\rowsp(Y),\rowsp(Z)\subseteq \rowsp(C)$. Thus, there exist
  $P,Q\in\F_q^{2\times 2}$ such that
  \begin{align*}
    Y&=PC, & Z&=QC.
  \end{align*}
  Since $C$ has full rank, there exists $W\in\F_q^{n\times 2}$ such that
  $CW=I_2$. Thus,
  \[
  I_2= CW = YW+ZW = PCW + QCW = P+Q.
  \]

  On the one hand, we trivially have
  \[ \rank P \geq \rank (PC).\]
  On the other hand, $P=(PC)W$, so
  \[ \rank P \leq \rank (PC).\]
  Similar statements hold for $Q$. Thus,
  \begin{align*}
    \rank P &= \rank (PC) = \rank Y = 1, \\
    \rank Q &= \rank (QC) = \rank Z = 1, \\
  \end{align*}
  and therefore
  \[ \det P = \det (I_2-P) = 0.\]
  For a $2\times 2$ matrix
  \[ \det (I_2-P) = 1-\tr P + \det P,\]
  so $\tr(P)=1$. By the Cayley-Hamilton Theorem,
  \[ P^2-\tr(P)P+\det(P)I_2 = 0,\]
  which in our case becomes
  \[ P^2=P.\]

  If $x\in \im P \cap \ker P$, then there exists $y\in\F_q^2$ such
  that $x=Py$. Then
  \[ 0 = P x = P(Py) = P^2 y = Py = x,\]
  so $ \im P \cap \ker P = \set{0}$, namely, $\im P$ and $\ker P$ are
  distinct one-dimensional subspaces. Finally, $P$ acts as the
  identity on $\im P$, so $\ker(I_2-P)=\im P$.
\end{proof}

We are now in a position to state and prove the main theorem of this section, that shows the conjecture holds for order $t=2$.

\begin{theorem}
  \label{th:t2}
  Let $\cC \subseteq \F_q^n$, $n\geq 2$, be a linear code with $\dim
  \cC\geq 2$. Then
  \[ d_2(\cC) \leq 2R_2(\cC)+2,\]
  and in particular,
  \[ \delta_2(\cC) \leq R_2(\cC).\]
\end{theorem}

\begin{proof}
  For convenience, define
  \begin{align*}
    d_2 &= d_2(\cC), &
    r_2 &= R_2(\cC).
  \end{align*}
  Assume to the contrary that
  \begin{equation}
    \label{eq:asscont}
    d_2\geq 2r_2+3.
  \end{equation}
  Choose $C\in \cC^2 \subseteq \F_q^{2\times n}$ such that
  $\wt(C)=d_2$ and $\rank(C)=2$ (here we use $\dim \cC\geq
  2$). Thus, $\abs{\supp(C)}=d_2$, and $C_{|j}\neq 0$ for every
  $j\in\supp(C)$.

  Each non-zero column $C_{|j}\in\F_q^2$ determines a projective point
  $\ang{C_{|j}}\in \PG(1,q)$. Apply Lemma~\ref{lem:partition} to the
  multiset $\mset{ \ang{C_{|j}} : j\in\supp(C) }$ to obtain a partition
  of $\supp(C)$ into two sets, $A$ and $B$, such that
  \begin{align*}
    A \cup B &= \supp(C), &
    \abs{A} &= \floor*{\frac{d_2}{2}}, &
    \abs{B} &= \ceil*{\frac{d_2}{2}}, &
  \end{align*}
  and for any two distinct projective points, $\ang{x},\ang{y}\in\PG(1,q)$,
  \[
  \#_{\ang{x}}(\mset{ \ang{C_{|j}} : j\in A }) + \#_{\ang{y}}(\mset{ \ang{C_{|j}} : j\in B }) \leq \floor*{\frac{d_2}{2}}+1.
  \]

  Define $X\in\F_q^{2\times n}$ in the following way:
  \[ X_{|j} =
  \begin{cases}
    C_{|j}, & j\in A, \\
    0, & \text{otherwise.}
  \end{cases}
  \]
  By the definition of the second-order covering radius, there exists
  $Y\in\cC^2$ such that
  \begin{equation}
    \label{eq:xminy}
    \wt(X-Y)\leq r_2.
  \end{equation}
  Define $Z\eqdef C-Y$. By the linearity of the code, $Z\in\cC^2$ as well.

  Let us consider the supports of $Y$ and $Z$. Since $Y=X-(X-Y)$, one has
  \[
  \supp(Y) \subseteq \supp(X) \cup\supp(X-Y),
  \]
  which implies
  \[
  \wt(Y) \leq \abs{A} + r_2 \leq \floor*{\frac{d_2}{2}}+r_2 < d_2,
  \]
  where the last inequality follows
  from~\eqref{eq:asscont}. Similarly, $Z=(C-X)+(X-Y)$, so
  \[
  \supp(Z) \subseteq \supp(C-X) \cup\supp(X-Y),
  \]
  and hence,
  \[
  \wt(Z) \leq \abs{B} + r_2 \leq \ceil*{\frac{d_2}{2}}+r_2 < d_2.
  \]

  By~\eqref{eq:defdt}, no matrix in $\cC^2$ of block
  weight strictly smaller than $d_2$ can have rank $2$. Thus, $\rank
  Y\leq 1$ and $\rank Z \leq 1$. However,
  \[ 2 = \rank C = \rank (Y+Z) \leq \rank Y + \rank Z \leq 2.
  \]
  It follows that
  \[ \rank Y = \rank Z = 1.\]

  By Lemma~\ref{lem:matrix}, there exists $P\in\F_q^{2\times 2}$ such
  that $Y=PC$, $Z=(I_2-P)C$, and $P^2=P$. Additionally, since $\dim
  \ker P = \dim \im P = 1$, all the non-zero vectors of $\im P$ are a
  single projective point $\ang{i}\in\PG(1,q)$, and similarly, all the
  non-zero vectors of $\ker P$ are a single projective point
  $\ang{k}\in\PG(1,q)$, and the two are distinct.

  We now consider the second-order
  covering radius of the code by measuring the distance of $X-Y$ to
  the code, coordinate by coordinate. If $j\in A$, then
  \[ (X-Y)_{|j} = C_{|j}- PC_{|j} = (I_2-P)C_{|j}.
  \]
  The $j$-th column vanishes exactly when $C_{|j}\in \ker(I_2 -P) =
  \im P$, by Lemma~\ref{lem:matrix}. If $j\in B$, then $X_{|j}=0$, and
  so
  \[ (X-Y)_{|j} = -PC_{|j},\]
  which vanishes exactly when $C_{|j}\in\ker P$. Finally, if $j\notin
  A\cup B$, then $C_{|j}=0$ and also $X_{|j}=0$. In addition,
  $Y_{|j}=PC_{|j}$ and so $(X-Y)_{|j}=0$. It follows that
  \begin{align*}
  \wt(X-Y) &= d_2-(\#_{\ang{i}}(A) + \#_{\ang{k}}(B))\\ &\overset{(a)}{\geq} d_2-\parenv*{\floor*{\frac{d_2}{2}}+1} = \ceil*{\frac{d_2}{2}}-1 \overset{(b)}{\geq} r_2+1,
  \end{align*}
  where $(a)$ follows from Lemma~\ref{lem:partition}, and $(b)$
  follows by~\eqref{eq:asscont}. However, this contradicts~\eqref{eq:xminy}.

  Thus,
  \[ d_2(\cC)\leq 2R_2(\cC)+2,\]
  and then for $\delta_2(\cC)$ we get
  \[ \delta_2(\cC) = \floor*{\frac{d_2(\cC)-1}{2}} \leq R_2(\cC).\]
\end{proof}

\section{Bounds by Rate}
\label{sec:rate}

In this section we change our approach, compared with the previous section. Instead of looking at a specific order, we look at codes with a certain range of rates. We show that codes of rate at most $3/5$ satisfy the conjecture for all relevant orders.

The main theorem requires several preparatory steps which we briefly outline. First, in Lemma~\ref{lem:Rtt}, we show the conjecture holds when $R_t(\cC)=t$. We then move on to show it holds for very high orders, $\max\set{1,n-k-4}\leq t\leq \min\set{k,n-k}$ in Lemma~\ref{lem:hight}. Another stepping stone is proved in Lemma~\ref{lem:step5t}, where low dimension is considered, $k\leq 5t-2$, which then leads to the main result in Theorem~\ref{th:rate35}.

\begin{lemma}
  \label{lem:Rtt}
  Let $\cC$ be an $[n,k]_q$ code. If $1\leq t\leq \min\set{k,n-k}$ and
  $R_t(\cC)=t$, then $d_t(\cC)\leq 2t+1$, and in particular,
  \[\delta_t(\cC)\leq R_t(\cC).\]
\end{lemma}
\begin{proof}
  If $n-k=t$, then by the generalized Singleton bound~\cite[Theorem
    1]{Wei91}, $d_t(\cC)\leq n-k+t=2t$, and the claim is proved. Let
  us therefore continue with the case $n-k\geq t+1$. Let
  $H\in\F_q^{(n-k)\times n}$ be parity-check matrix for $\cC$, hence
  $\rank H=n-k$.

  Choose an arbitrary $(t+1)$-dimensional subspace $W\subseteq
  \F_q^{n-k}$, and denote by $N_W$ the number of non-zero columns of
  $H$ that appear in $W$. Every $t$-dimensional subspace $T\subseteq
  W$ is spanned by some $t$ columns of $H$, thus, $T$ contains at
  least $t$ non-zero columns from $H$. By straightforward counting
  there are $(q^{t+1}-1)/(q-1)$ such subspaces $T$ of $W$, and any
  non-zero vector in $W$ is contained in exactly $(q^t-1)/(q-1)$ of
  them. It follows that
  \[ N_W \geq \ceil*{t \frac{q^{t+1}-1}{q^t-1}} \geq qt+1 \geq 2t+1.\]
  Choose $2t+1$ of these non-zero columns of $H$ in $W$. Placing these
  in an $(n-k)\times (2t+1)$ matrix denoted $H'$, we have $\rank
  H'\leq t+1$, and thus as a parity-check matrix for a code $\cC'$,
  the code has dimension at least $\dim \cC'\geq t$. But that means
  there is a subcode of $\cC$, of dimension $t$, whose support is
  restricted to the $2t+1$ chosen coordinates, namely, $d_t(\cC)\leq
  2t+1$.
\end{proof}

We now move on to consider the conjecture for high orders.

\begin{lemma}
  \label{lem:hight}
  Let $\cC$ be an $[n,k]_q$ code. For all $\max\set{1,n-k-4}\leq t\leq
  \min\set{k,n-k}$ we have
  \[ \delta_t(\cC)\leq R_t(\cC).\]
\end{lemma}

\begin{proof}
  From~\cite[Eq. (1)]{EliFirSch21a}, we have $R_t(\cC)\geq t$. If
  $R_t(\cC)=t$, then by Lemma~\ref{lem:Rtt} the claim is
  proved. Otherwise, we have $R_t(\cC)\geq t+1$.

  From Lemma~\ref{lem:knownsingleton} we have
  \[ d_t(\cC)\leq n-k+t \leq 2t+4.\]
  Hence,
  \[\delta_t(\cC) = \floor*{\frac{d_t(\cC)-1}{2}}\leq t+1 \leq R_t(\cC).\]
\end{proof}

\begin{corollary}
  Let $\cC$ be an $[n,k]_q$ code, with redundancy $n-k\leq 7$.
  Then for all $1\leq t\leq \min\set{k,n-k}$ we have
  \[ \delta_t(\cC)\leq R_t(\cC).\]
\end{corollary}
\begin{proof}
  For $t=1$ this known as the classical packing radius vs. the
  covering radius of the code. For $t=2$ this is proved in
  Theorem~\ref{th:t2}. For $t\geq 3$ we have $t\geq n-k-4$, which is
  covered in Lemma~\ref{lem:hight}.
\end{proof}

The final stepping stone before the main theorem of this section considers low-dimension codes with respect to the conjecture.

\begin{lemma}
\label{lem:step5t}
  Let $\cC$ be an $[n,k]_q$ code. If $3\leq t\leq \min\set{k,n-k}$,
  and $k\leq 5t-2$, then
  \[ \delta_t(\cC) \leq R_t(\cC).\]
\end{lemma}

\begin{proof}
  If $n-k\leq t+4$, then the claim is already proved in
  Lemma~\ref{lem:hight}. Thus, let us consider $n-k \geq t+5$.

Define
  \[ r_{n-k} \eqdef\floor*{\frac{n-k+t-3}{2}}.\]
  Then by Lemma~\ref{lem:knownsingleton},
  \[
  \delta_t(\cC) = \floor*{\frac{d_t(\cC)-1}{2}}
  \leq \floor*{\frac{n-k+t-1}{2}} = r_{n-k}+1.
  \]
  We notice that now it is enough to show that
  \begin{equation}
    \label{eq:need}
    \cV_{q^t}(n,r_{n-k})<q^{t(n-k)},
  \end{equation}
  because by~\eqref{eq:nec} it would imply $R_t(\cC) \geq r_{n-k}+1 \geq
  \delta_t(\cC)$, as claimed.

  We proceed by induction on the redundancy, $n-k$, in the following
  manner. The base case $n-k=t+5$ will be proved separately, as the
  induction base. Then, for the induction step we will prove two
  claims. We note that for the base case, $n-k+t-3=2t+2$ is
  even. Thus, for the first induction step we will increase $n-k$ by
  $1$ (by increasing $n$ by $1$), which does not cause a change in
  $r_{n-k}$ due to the floor operation, i.e., $r_{n-k+1}=r_{n-k}$. The
  second induction step we take is increasing $n-k$ by $2$, which
  increases the radius $r_{n-k}$ by $1$, namely,
  $r_{n-k+2}=r_{n-k}+1$. For technical reasons, we not only
  prove~\eqref{eq:need} in this induction, but also an auxiliary claim, that
  whenever $n-k+t-3$ is even,
  \begin{equation}
    \label{eq:need2}
    n \leq q^t(r_{n-k}+1)-1.
  \end{equation}

  \textbf{Induction Base:} The redundancy is $n-k=t+5$. In this case,
  $r_{n-k}=t+1$. By the lemma requirement of $k\leq 5t-2$ we obtain $n\leq
  6t+3$. Since $\cV_{q^t}(n,r_{n-k})$ is increasing in $n$, it suffices to
  show
  \begin{equation}
    \label{eq:almost}
    \cV_{q^t}(6t+3,t+1) < q^{t(t+5)}.
  \end{equation}
  We generally note that
  \begin{equation}
  \label{eq:upball}
  \cV_{q^t}(n,r) \leq \binom{n}{r}q^{tr},
  \end{equation}
  since the RHS counts the ways to choose a subset of $r$ positions
  out of $n$, filling it in any way from an alphabet of size $q^t$,
  and setting $0$ outside this subset. This obviously over-counts the
  volume of the ball $\cV_{q^t}(n,r)$. Thus, if we show for $t\geq 3$ that
  \[ \binom{6t+3}{t+1} < q^{4t},\]
  it would suffice to prove~\eqref{eq:almost}. We show this is indeed
  satisfied for $q\geq 3$ or $t\geq 6$, i.e., except for
  $(q,t)\in\set{(2,3),(2,4),(2,5)}$. These remaining three cases will
  be checked by inspection against~\eqref{eq:almost}.

  Write $b_t\eqdef \binom{6t+3}{t+1}$. We first show that $b_t/16^t$
  is strictly decreasing for $t\geq 3$. Indeed
  \[
  \frac{b_{t+1}}{b_t}= \frac{6t+9}{t+2}\prod_{j=0}^4 \frac{6t+4+j}{5t+3+j}.
  \]
  For $t\geq 3$, the factor outside the product is less than $6$, the
  first two factors in the product are at most $11/9$, and the last
  three are at most $6/5$. Consequently,
  \begin{equation}
    \label{eq:bratio}
  \frac{b_{t+1}}{b_t} < 6\parenv*{\frac{11}{9}}^2 \parenv*{\frac{6}{5}}^3
  = \frac{1936}{125}<16.
  \end{equation}
  For $q=2$,
  \[ b_6 = \binom{39}{7}=15380937 < 16777216 = 2^{4\cdot 6},\]
  and~\eqref{eq:bratio} shows~\eqref{eq:almost} holds for all $q=2$
  and $t\geq 6$. For $q\geq 3$,
  \[
  b_3 = \binom{21}{4} = 5985 < 3^{4\cdot 3} \leq q^{4\cdot 3},
  \]
  and~\eqref{eq:bratio} shows~\eqref{eq:almost} holds for all $q\geq
  3$ and $t\geq 3$. The only remaining cases are
  $(q,t)\in\set{(2,3),(2,4),(2,5)}$. In these cases, \eqref{eq:almost} is
  checked by inspection,
  \begin{align*}
    \cV_{2^3}(6\cdot 3+3,3+1) = 14836613 &< 16777216 = 2^{3(3+5)}, \\
    \cV_{2^4}(6\cdot 4+3,4+1) = 62202763756 &< 68719476736 = 2^{4(4+5)}, \\
    \cV_{2^5}(6\cdot 5+3,5+1) = 989803358666992 &< 1125899906842624 = 2^{5(5+5)}. 
  \end{align*}

  To complete the induction base, we also need to verify the auxiliary
  claim~\eqref{eq:need2}. Indeed
  \begin{align*}
    q^t(r_{n-k} + 1 ) - 1 & \overset{(a)}{\geq} 8(t+1+1)-1 = 8t+15 \overset{(b)}{\geq} 6t+3 \overset{(c)}{\geq} k+t+5 = n,
  \end{align*}
  where $(a)$ follows from $r_{n-k}=t+1$ and $q^t\geq 2^3$, $(b)$
  holds for $t\geq 1$, and $(c)$ follows from the requirement $k\leq 5t-2$.

  \textbf{Induction Hypothesis:} Assume $n-k+t-3$ is even,
  and~\eqref{eq:need} as well as~\eqref{eq:need2}, hold.

  \textbf{First Induction Step:} Increasing the length $n$ by $1$. We
  first note that $r_{n-k+1}=r_{n-k}$. We use the
  well-known binomial identity, $\binom{a}{b} = \binom{a-1}{b} +
  \binom{a-1}{b-1}$. Using this with the expression for the ball size,
  \begin{align*}
    \cV_{q^t}(n+1,r_{n-k+1}) &= \cV_{q^t}(n+1,r_{n-k}) =
    \cV_{q^t}(n,r_{n-k})+ (q^t-1)\cV(n,r_{n-k}-1) \\
    & \overset{(a)}{\leq} q^t \cV_{q^t}(n,r_{n-k})
    \overset{(b)}{<} q^t q^{t(n-k)} = q^{t(n-k+1)},
  \end{align*}
  where $(a)$ follows from simple monotonicity
  $\cV_{q^t}(n,r_{n-k}-1)\leq \cV_{q^t}(n,r_{n-k})$, and $(b)$ follows
  from the induction hypothesis. Thus, \eqref{eq:need} holds when $n$
  is increased by $1$, and this induction step is proved. We also note
  that we do not need to prove the auxiliary claim in this case, since
  $(n+1)-k+t-3$ is odd.

  \textbf{Second Induction Step:} Increasing the length by $2$. Now we
  have $r_{n-k+2} = r_{n-k}+1$. Using the same strategy as above,
  \begin{align*}
    \cV_{q^t}(n+2,r_{n-k+2}) &= \cV_{q^t}(n+2,r_{n-k}+1) \\
    & = \cV_{q^t}(n,r_{n-k}+1) + 2(q^t-1)\cV_{q^t}(n,r_{n-k}) \\
    & \qquad +(q^t-1)^2\cV_{q^t}(n,r_{n-k}-1).
  \end{align*}
  We observe that
  \begin{align*}
    \cV_{q^t}(n,r_{n-k}+1) &= \cV_{q^t}(n,r_{n-k}) + \binom{n}{r_{n-k}+1}(q^t-1)^{r_{n-k}+1}, \\
    \cV_{q^t}(n,r_{n-k}-1) &= \cV_{q^t}(n,r_{n-k}) - \binom{n}{r_{n-k}}(q^t-1)^{r_{n-k}}. \\
  \end{align*}
  Using this in the above,
  \begin{align*}
    \cV_{q^t}(n+2,r_{n-k+2}) &= q^{2t}\cV_{q^t}(n,r_{n-k}) + \binom{n}{r_{n-k}+1}(q^t-1)^{r_{n-k}+1} \\
    & \qquad - \binom{n}{r_{n-k}}(q^t-1)^{r_{n-k}+2} \\
    & < q^{t(n-k+2)}+ q^{r_{n-k}+1}\parenv*{\binom{n}{r_{n-k}+1} - \binom{n}{r_{n-k}}(q^t-1)},
  \end{align*}
  where the inequality is by the induction hypothesis. We would like
  to show that the parentheses contain a non-positive value, i.e.,
  \[
  \binom{n}{r_{n-k}+1} - \binom{n}{r_{n-k}}(q^t-1) \leq 0.
  \]
  Equivalently, after rearranging, we want to show
  \[ n \leq q^t(r_{n-k}+1) -1.\]
  But that is guaranteed by the induction hypothesis on~\eqref{eq:need2}. Thus,
  \[ \cV_{q^t}(n+2,r_{n-k+2}) < q^{t(n-k+2)},\]
  which completes the induction step for claim~\eqref{eq:need}. We also need
  to prove the induction step for the auxiliary~\eqref{eq:need2}. We have
  \[
  n+2 \overset{(a)}{\leq} q^t(r_{n-k}+1)+1 = q^t(r_{n-k+2}+1)-q^t+1 \overset{(b)}{\leq}
  q^t(r_{n-k+2}+1)-1
  \]
  where $(a)$ is in the induction hypothesis on~\eqref{eq:need2}, and
  $(b)$ follows from $q^t\geq 2$. This completes the induction and the
  proof.
\end{proof}

In the following theorem we shall be manipulating the volume of a ball, and so we present two simple tools. We assume an alphabet of size $Q\geq 2$ and ball radius $r$. Let $0<x<1$ be a real number. Then,
\begin{equation}
\label{eq:generating}
\begin{split}
x^r \cV_{Q}(n,r) &= \sum_{j=0}^r\binom{n}{j}(Q-1)^j x^r
\leq \sum_{j=0}^r\binom{n}{j}(Q-1)^j x^j \\
&\leq \sum_{j=0}^n\binom{n}{j}(Q-1)^j x^j 
= (1+(Q-1)x)^n.
\end{split}
\end{equation}
Next, we want to show monotonicity of $\cV_Q(n,r)/Q^{n-k}$ in the
alphabet size $Q$. To that end, treat $Q$ as a real number, $Q>1$, and define $f_j(Q)\eqdef(Q-1)^j/Q^{n-k}$. We observe that
\[
\frac{f'_j(Q)}{f_j(Q)}=\frac{j}{Q-1}-\frac{n-k}{Q} = \frac{(n-k)-(n-k-j)Q}{Q(Q-1)}.
\]
This is negative whenever $(n-k-j)Q>n-k$, and since the denominator is positive, that means $f_j(Q)$ is decreasing. If $(n-k-j)Q>n-k$ for all $0\leq j\leq r$, then $\cV_Q(n,r)/Q^{n-k} = \sum_{j=0}^r \binom{n}{j} f_j(Q)$ is decreasing in the alphabet size $Q$.

We now state and prove the main theorem, showing that codes of rate at most $3/5$ satisfy the conjecture for all $1\leq t\leq \min\set{k,n-k}$.

\begin{theorem}
\label{th:rate35}
  Let $\cC$ be an $[n,k]_q$ code. If $k/n \leq 3/5$, then for every $1\leq t\leq \min\set{k,n-k}$,
  \[ \delta_t(\cC)\leq R_t(\cC).\]
\end{theorem}

\begin{proof}
For $t=1$ the claim is already known, and for $t=2$ it is solved in Theorem~\ref{th:t2}, so we assume $t\geq 3$. By Lemma~\ref{lem:step5t}, the claim holds when $k\leq 5t-2$. Thus, in the remainder of the proof we assume $k\geq 5t-1$. By our requirement that $k/n \leq 3/5$ this means $k\leq 3(n-k)/2$, so
\begin{equation}
\label{eq:nkreq}
n-k\geq \frac{2(5t-1)}{3}.
\end{equation}
The proof proceeds by using the ball-covering bound. We define $r_{n-k}$ as in the proof of Lemma~\ref{lem:step5t}, and so to prove our claim it suffices to show~\eqref{eq:need} holds. We distinguish between cases depending on the value of $t$.

\textbf{Case 1:} $t\geq 6$. We observe
\[ n+t r_{n-k} \leq \frac{5(n-k)}{2}+\frac{t(n-k+t-3)}{2}= t(n-k)-\frac{(t-5)(n-k)-t(t-3)}{2}.\]
We contend that in this range, together with~\eqref{eq:nkreq},
\[
(t-5)(n-k)-t(t-3) \geq (t-5)\frac{2(5t-1)}{3}-t(t-3)=\frac{7t^2-43t+10}{3}>0
\]
for all $t\geq 6$. Thus,
\[ n+t r_{n-k} < t(n-k).\]
But now,
\begin{align*}
\cV_{q^t}(n,r_{n-k}) = \sum_{i=0}^{r_{n-k}} \binom{n}{i}(q^t-1)^i < 2^n q^{t r_{n-k}} \leq q^{n+tr_{n-k}} < q^{t(n-k),}
\end{align*}
and~\eqref{eq:need} holds, so $\delta_t(\cC)\leq R_t(\cC)$.

\textbf{Case 2:} $t=4,5$.
These cases require a slightly
more nuanced approach for bounding $\cV_{q^t}(n,r_{n-k})$. By~\eqref{eq:generating}, we choose $x=\frac{1}{4(q^t-1)}$, and obtain
\[ \cV_{q^t}(n,r_{n-k}) \leq (4(q^t-1))^{r_{n-k}}\parenv*{\frac{5}{4}}^n.\]
Squaring this, and dividing by $q^{2t(n-k)}$ we get
\[
\parenv*{\frac{\cV_{q^t}(n,r_{n-k})}{q^{t(n-k)}}}^2
\leq
\frac{(4(q^t-1))^{2r_{n-k}}(5/4)^{2n}}{q^{2t(n-k)}}.
\]
We now use the fact that $2r_{n-k}\leq n-k+t-3$ and $2n\leq 5(n-k)$,
\begin{equation}
\label{eq:genq}
\parenv*{\frac{\cV_{q^t}(n,r_{n-k})}{q^{t(n-k)}}}^2
\leq
(4(q^t-1))^{t-3} \parenv*{\frac{(5/4)^5 4(q^t-1)}{q^{2t}}}^{n-k}.
\end{equation}

Let us first see what happens for $q=2$. When $t=4$, \eqref{eq:genq} becomes
\begin{align*}
\parenv*{\frac{\cV_{2^4}(n,r_{n-k})}{2^{4(n-k)}}}^2
&\leq
(4(2^4-1))^{4-3} \parenv*{\frac{(5/4)^5 4(2^4-1)}{2^{2\cdot 4}}}^{n-k} \\
&= 60 \parenv*{\frac{46875}{65536}}^{n-k} \overset{(a)}{\leq} 60 \parenv*{\frac{46875}{65536}}^{13} < 1,
\end{align*}
where $(a)$ follows from~\eqref{eq:nkreq}. We aim to use the monotonicity in $Q$ discussed before stating Theorem~\ref{th:rate35}. To that end, assume $j\leq r_{n-k} = \floor{\frac{n-k+t-3}{2}}\leq \frac{n-k+1}{2}$, since $t=4$. Additionally, $Q\geq 2^4$ in this case. Then
\[
(n-k-j)Q \geq 8(n-k-1) > n-k,
\]
where the last inequality follows using~\eqref{eq:nkreq}. Thus, for $t=4$, for all $q\geq 2$
\[
\parenv*{\frac{\cV_{q^4}(n,r_{n-k})}{q^{4(n-k)}}}^2 \leq
\parenv*{\frac{\cV_{2^4}(n,r_{n-k})}{2^{4(n-k)}}}^2 < 1.
\]
Hence, \eqref{eq:need} holds and $\delta_4(\cC)\leq R_4(\cC)$.

For $t=5$ the proof is essentially the same. We take $q=2$ first,
\begin{align*}
\parenv*{\frac{\cV_{2^5}(n,r_{n-k})}{2^{5(n-k)}}}^2
&\leq
(5(2^5-1))^{5-3} \parenv*{\frac{(5/4)^5 4(2^5-1)}{2^{2\cdot 5}}}^{n-k} \\
&\leq  24025 \parenv*{\frac{96875}{262144}}^{16}< 1.
\end{align*}
Monotonicity in the alphabet size is similar, since we now have $j\leq \frac{n-k+2}{2}$, and then
\[
(n-k-j)Q \geq 16(n-k-2) > n-k. 
\]
Thus,
\[
\parenv*{\frac{\cV_{q^5}(n,r_{n-k})}{q^{5(n-k)}}}^2 \leq
\parenv*{\frac{\cV_{2^5}(n,r_{n-k})}{2^{5(n-k)}}}^2 < 1,
\]
and therefore $\delta_5(\cC)\leq R_5(\cC)$.

\textbf{Case 3:} $t=3$. This case is further subdivided into $q\geq 3$ and $q=2$.

We start with $t=3$ and $q\geq 3$, which proceeds along the same lines as the case $t=4$. We have $r_{n-k}=\floor{\frac{n-k}{2}}$. Taking $q=3$,
\begin{align*}
\parenv*{\frac{\cV_{3^3}(n,r_{n-k})}{3^{3(n-k)}}}^2
&\leq
\parenv*{\frac{(5/4)^5 4(3^3-1)}{3^{2\cdot 3}}}^{n-k}
\leq \parenv*{\frac{40625}{93312}}^{10} < 1.
\end{align*}
For monotonicity, $Q\geq 3^3=27$, and for $j\leq \frac{n-k}{2}$ we have
\[
(n-k-j)Q \geq 27\frac{n-k}{2} > n-k,
\]
so
\[
\parenv*{\frac{\cV_{q^3}(n,r_{n-k})}{q^{3(n-k)}}}^2 \leq
\parenv*{\frac{\cV_{3^3}(n,r_{n-k})}{3^{3(n-k)}}}^2 < 1,
\]
and therefore $\delta_3(\cC)\leq R_3(\cC)$ for $q\geq 3$.

Our final case is $q=2$. Our previous strategy is not immediately helpful since the fraction inside the parentheses on the RHS of~\eqref{eq:genq} is greater than $1$. We therefore make the following steps. We first recall that $k\geq 5t-1 = 14$. Let us get a better handle on $d_3(\cC)$: by using Lemma~\ref{lem:knownvla}, and choosing $r=3$, $s=k$ (while remembering $d_k(\cC)\leq n$), $\ell=k-h$ and $3\leq h\leq k$, we get
\begin{equation}
\label{eq:d3opt}
d_3(\cC) \leq \floor*{\frac{2^h-2^{h-3}}{2^h-1}(n-k+h)}.
\end{equation}
We can choose any $3\leq h\leq 14$, and choosing $h=6$ will suffice, as we shall see. With this choice,
\[
d_3(\cC)\leq \floor*{\frac{8}{9}(n-k+6)}.
\]
To use the same line of argument as before, we choose
\[
r^*_{n-k} \eqdef \floor*{\frac{d_3(\cC)-3}{2}} \leq \frac{4(n-k)}{9}+\frac{7}{6}.
\]
We apply~\eqref{eq:generating} with $x=\frac{1}{32}$ to get
\begin{align*}
\cV_{2^3}(n,r^*_{n-k}) &\leq 32^{r^*_{n-k}} \parenv*{\frac{39}{32}}^n \leq 32^{4(n-k)/9 + 7/6} \parenv*{\frac{39}{32}}^{5(n-k)/2} \\
&= 2^{3(n-k)} \parenv*{2^{35/6}\parenv*{\frac{39^{5/2}}{2^{239/18}}}^{n-k}}.
\end{align*}
For $n-k\geq 91$, the expression in the parentheses on the RHS is strictly smaller than $1$, and hence, $\delta_3(\cC)\leq R_3(\cC)$.

We are only missing the cases $10\leq n-k\leq 90$. By our restriction on the rate $n/k\leq 3/5$, we have $n\leq \floor{\frac{5}{2}(n-k)}\eqdef N_{n-k}$. As for bounding $d_3(\cC)$, we can use the same argument resulting in~\eqref{eq:d3opt}, though we can choose $3\leq h\leq k$, where we have already seen that $k\geq 14$. It will suffice (as we shall see below) to choose $3\leq h\leq 6$, and obtain
\begin{align*}
& d_3(\cC) \\
&\quad\leq \min\set*{
n-k+3,
\floor*{\frac{14(n-k+4)}{15}},
\floor*{\frac{28(n-k+5)}{31}},
\floor*{\frac{8(n-k+6)}{9}}
} \\
&\quad \eqdef D_{n-k}.
\end{align*}
Finally, set
\[
r^*_{n-k} \eqdef \floor*{\frac{D_{n-k}-3}{2}}.
\]
If we can show
\[
\cV_{2^3}(N_{n-k},r^*_{n-k}) < 2^{3(n-k)},
\]
then by~\eqref{eq:nec} we have $\delta_3(\cC)\leq R_3(\cC)$. Since there are a finite number of cases for the parameter $n-k$, i.e., $10\leq n-k\leq 90$, these may be done by a finite (tedious) inspection which shows the requirement does hold.
\end{proof}

\section{Asymptotic Results}
\label{sec:asymptotic}

This section considers the conjecture from a different point of view. This time, we do not restrict the order $t$, nor the rate $k/n$. Instead, we restrict the length of the code. The main result in Theorem~\ref{th:asymptotic} shows that the conjecture holds for any rate $0<\alpha<1$, and order $1\leq t\leq \min\set{k,n-k}$, provided the code length $n\geq N_\alpha$, where $N_\alpha$ depends on $\alpha$ only.

Like in the previous section, reaching the main result requires a sequence of lemmas considering smaller cases. In Lemma~\ref{lem:highrad} we show the conjecture holds when the radius is high enough. Then, Lemma~\ref{lem:affball} upper bounds the generalized Hamming weight of a code. Lemma~\ref{lem:gap} is a purely technical lemma showing a gap exists between two functions that will be needed later. Lemma~\ref{lem:uniform} nearly proves the main claim, only the threshold does not depend only on the rate, but also on the alphabet size $q$ and the order $t$. All of these enable the proof of the main claim.

\begin{lemma}
\label{lem:highrad}
Let $\cC$ be an $[n,k]_q$ code, and let $1\leq t\leq \min\set{k,n-k}$. If $R_t(\cC)\geq 2(n-k)/3$ then
\[
\delta_t(\cC) \leq R_t(\cC).
\]
\end{lemma}
\begin{proof}
If $t=1$ then the claim is the classic packing vs.~covering radii. For $t=2$ the claim holds by Theorem~\ref{th:t2}. Assume $t\geq 3$. If $R_t(\cC)=t$, then the claim holds by Lemma~\ref{lem:Rtt}. Let us therefore assume $R_t(\cC)\geq t+1$. Assume to the contrary that 
\begin{equation}
    \label{eq:contrar}
d_t(\cC)\geq 2R_t(\cC)+3.
\end{equation}
Define
\begin{align*}
\ell &\eqdef 2R_t(\cC)+2-t, &
m &\eqdef n-k-\ell.
\end{align*}
By~\eqref{eq:contrar} and Lemma~\ref{lem:knownsingleton}, $\ell\leq n-k-1$, and thus $m\geq 1$. Additionally, $\ell\geq R_t(\cC)+3\geq t+1+3\geq 7$.

Let $H\in\F_q^{(n-k)\times n}$ be a parity-check matrix for $\cC$. We contend that every $\ell$-dimensional subspace of $\F_q^{n-k}$ contains at most $2R_t(\cC)+1$ columns of $H$ (counted with multiplicity). Indeed, if $2R_t(\cC)+2$ such columns are present, we can form them into a matrix $H'$ of size $(n-k)\times (2R_t(\cC)+2)$ and whose kernel has dimension at least $2R_t(\cC)+2-\ell = t$. This implies a $t$-dimensional subcode of $\cC$ whose support is at most $2R_t(\cC)+2$, contradicting~\eqref{eq:contrar}.

Every vector in $\F_q^{n-k}\setminus\set{0}$ appears in a $\frac{q^\ell-1}{q^{n-k}-1}$-fraction of the $\ell$-dimensional subspaces mentioned above. Thus, each such subspaces contributes at most $2R_t(\cC)+1$ columns to $H$, and thus
\begin{equation}
\label{eq:upn}
n \leq (2R_t(\cC)+1) \frac{q^{n-k}-1}{q^{\ell}-1} < (2R_t(\cC)+1)\frac{q^m}{1-q^{-\ell}}.
\end{equation}

Recalling~\eqref{eq:downball} and~\eqref{eq:upball}, we have
\[
q^{t(n-k)} \leq \cV_{q^t}(n,R_t(\cC)) \leq \binom{n}{R_t(\cC)}q^{tR_t(\cC)}.
\]
Comparing both sides gives
\begin{equation}
\label{eq:temp1}
q^{t(n-k-R_t(\cC))} \leq \binom{n}{R_t(\cC)}.
\end{equation}
Using Stirling's approximation (e.g., see~\cite[p.~309, Eq.~(17)]{MacSlo78}), for all $r\geq 1$,
\[
r! > \sqrt{2\pi r} \parenv*{\frac{r}{e}}^r \geq \parenv*{\frac{r}{e}}^r e^{1/2} > \parenv*{\frac{r}{e}}^r \parenv*{1+\frac{1}{2r}}^r = \parenv*{\frac{r+\frac{1}{2}}{e}}^r.
\]
Thus,
\begin{align*}
\binom{n}{R_t(\cC)} &= \frac{n(n-1)\dots (n-R_t(\cC)+1)}{R_t(\cC)!} \\
&< \parenv*{\frac{en}{R_t(\cC)+\frac{1}{2}}}^{R_t(\cC)} < \parenv*{\frac{2eq^m}{1-q^{-\ell}}}^{R_t(\cC)},
\end{align*}
where the last inequality uses~\eqref{eq:upn}. Together with~\eqref{eq:temp1} we obtain
\begin{equation}
\label{eq:temp2}
q^{\frac{t(n-k-R_t(\cC))}{R_t(\cC)}-m} < \frac{2e}{1-q^{-\ell}}.
\end{equation}
Checking the exponent on the LHS,
\begin{align*}
\frac{t(n-k-R_t(\cC))}{R_t(\cC)}-m &= 2+(R_t(\cC)-t)\parenv*{2-\frac{n-k}{R_t(\cC)}} \\
&\overset{(a)}{\geq} 2 +\frac{R_t(\cC)-t}{2} \overset{(b)}{\geq} \frac{5}{2},
\end{align*}
where $(a)$ follows from the requirement that $R_t(\cC)\geq 2(n-k)/3$, and $(b)$ follows from $R_t(\cC)\geq t+1$.
But now, using $q\geq 2$ and $\ell\geq 7$,
\[
q^{\frac{t(n-k-R_t(\cC))}{R_t(\cC)}-m} \geq 2^{5/2}> \frac{2e}{1-2^{-7}} \geq \frac{2e}{1-q^{-\ell}},
\]
we contradict~\eqref{eq:temp2}. Then $d_t(\cC)\leq 2R_t(\cC)+2$ and the claim is proved.
\end{proof}

\begin{corollary}
\label{cor:log}
Let $\cC$ be an $[n,k]_q$ code, and let $1\leq t\leq \min\set{k,n-k}$. If $t(n-k)\log_2 q \geq 3n$, then
\[
\delta_t(\cC)\leq R_t(\cC).
\]
\end{corollary}
\begin{proof}
By~\eqref{eq:downball} and~\eqref{eq:upball},
\[
q^{t(n-k)} \leq \cV_{q^t}(n,R_t(\cC)) \leq \binom{n}{R_t(\cC)}q^{tR_t(\cC)}\leq 2^n q^{tR_t(\cC)}.
\]
Thus,
\[
R_t(\cC) \geq n-k-\frac{n\log_q 2}{t} \geq \frac{2(n-k)}{3},
\]
where the last inequality follows from the requirement. We use Lemma~\ref{lem:highrad} to complete the proof.
\end{proof}

The next lemma shows a straightforward upper bound on the generalized Hamming weight of a code.

\begin{lemma}
\label{lem:affball}
Let $\cC$ be an $[n,k]_q$ code. If $1\leq t\leq k$ and $0\leq r\leq n$ such that $\cV_q(n,r)> q^{n-k+t-1}$, then $d_t(\cC)\leq (t+1)r$.
\end{lemma}
\begin{proof}
Let us denote by $\cB_{q,r}(v)$ the Hamming ball of radius $r$ centered around $v\in\F_q^n$, i.e.,
\[
\cB_{q,r} \eqdef \set*{ v'\in \F_q^n : d(v,v')\leq r}.
\]
Simple averaging gives us
\[
\frac{1}{q^n} \sum_{v\in\F_q^n} \abs*{ \cC \cap \cB_{q,r}(v)} = q^{-(n-k)} \cV_q(n,r) > q^{t-1}.
\]
Thus, there exists some ball $\cB_{q,r}(v)$ that contains $t+1$ affinely independent vectors $c_0,\dots,c_t\in\cB_{q,r}(v)\cap \cC$, i.e., $c_1-c_0,c_2-c_0,\dots,c_t-c_0$ are linearly independent. We now observe that
\[
\bigcup_{i=1}^t \supp(c_i-c_0) \subseteq \bigcup_{i=0}^t \supp(c_i-v),
\]
and so the span of $c_1-c_0,\dots,c_t-c_0$ is a $t$-dimensional subcode of $\cC$ with support of size at most $(t+1)r$.
\end{proof}

For what follows, define
\[
\xi_{q,t}(s) \eqdef \min\set*{s,(t+1)H_q^{-1}(s)}.
\]
We shall need to show a gap between $\xi_{q,t}(s)$ and $H_{q^t}^{-1}(s)$. This is proved by the following technical lemma.

\begin{lemma}
\label{lem:gap}
For every $q\geq 2$, $t\geq 3$, and real number $0< s\leq 1$,
\[
\frac{1}{2}\xi_{q,t}(s) < H_{q^t}^{-1}(s).
\]
\end{lemma}
\begin{proof}
For convenience, define $\tau\eqdef \frac{t+1}{2}$, so $\tau\geq 2$. Consider the function over $(0,1)$
\[
\phi(x) \eqdef \frac{-x\ln x - (1-x)\ln(1-x)}{x}.
\]
We note that its second part, $-((1-x)/x)\ln(1-x)$ is strictly decreasing (since its derivative is $(\ln(1-x)+x)/x^2<0$). Thus, if $\tau x < 1$, 
\[
\phi(\tau x) < \phi(x)-\ln \tau.
\]
Additionally,
\[
\frac{q^t-1}{q-1} = 1+q+\dots+q^{t-1} \leq 2q^{t-1} \leq \tau q^{t-1}.
\]
Combining the above, it now follows that
\begin{align*}
t H_{q^t}(\tau x) < \tau H_q(x) +\tau x \log_q\parenv*{\frac{q^t-1}{\tau(q-1)}} \leq \tau H_q(x) + \tau x(t-1).
\end{align*}
The ratio $H_q(x)/x$ is strictly decreasing in $(0,1-q^{-1})$. Hence, there exists a unique $x_0\in (0,1-q^{-1})$ such that
\[
H_q(x_0) = (t+1) x_0 \eqdef s_0.
\]
If $s\leq s_0$, and setting $x=H_q^{-1}(s)$, then $\tau x \leq \frac{s}{2}\leq \frac{1}{2}$ and $\frac{1}{2}\xi_{q,t}(s)=\tau x$. Therefore
\[
t H_{q^t}(\xi_{q,t}(s)/2)=tH_{q^t}(\tau x) < \tau s + \tau x(t-1) \leq \tau s+ \frac{s(t-1)}{2}=ts.
\]
If $s\geq s_0$ then $\xi_{q,t}(s)=s$. Since $H_{q^t}(x)/x$ is decreasing,
\[
\frac{H_{q^t}(\xi_{q,t}(s)/2)}{s}= \frac{H_{q^t}(s/2)}{s} \leq \frac{H_{q^t}(s_0/2)}{s_0} < 1.
\]
In both cases $H_{q^t}(\xi_{q,t}(s)/2)<s$, which is equivalent to the claim.
\end{proof}

The next lemma nearly reaches the desired goal. It shows the conjecture holds for all long-enough codes. However, the threshold function depends on the rate, the alphabet size, and the order.

\begin{lemma}
\label{lem:uniform}
For fixed $q\geq 2$, $t\geq 3$, and real $0<\beta <1$, there exists $N(q,t,\beta)$ such that every $[n,k]_q$ code $\cC$ with $n\geq N(q,t,\beta)$, rate $\frac{k}{n}\leq 1-\beta$, and $t\leq \min\set{k,n-k}$, satisfies 
\[
\delta_t(\cC)\leq R_t(\cC).
\]
\end{lemma}

\begin{proof}
By Lemma~\ref{lem:gap}, the function $H_{q^t}^{-1}(s)-\frac{1}{2}\xi_{q,t}(s)$ is strictly positive, and by definition, continuous on the interval $[\beta,1]$. Thus, we can define
\[
\gamma \eqdef \min_{s\in [\beta,1]}\parenv*{H_{q^t}^{-1}(s)-\frac{1}{2}\xi_{q,t}(s)} > 0.
\]
As in the proof of Lemma~\ref{lem:gap}, let $x_0$ be the unique real number in $(0,1-q^{-1})$ such that $H_q(x_0)=(t+1)x_0$. Define
\[
\zeta \eqdef \min \set*{\frac{\gamma}{2(t+1)},\frac{1-q^{-1}-x_0}{2}}>0,
\]
as well as
\[
\Delta\eqdef H_q(x_0+\zeta)-H_q(x_0)>0,
\]
where positivity is deduced from the fact that $x_0+\zeta < 1-q^{-1}$, and the fact that $H_q(\cdot)$ is increasing in $(0,1-q^{-1})$.
We now contend that
\[
N(q,t,\beta)\eqdef \ceil*{\max\set*{\frac{1}{\zeta},\frac{4}{\Delta^2(\ln q)^2},\frac{2t}{\Delta}}}
\]
satisfies the claim.

Before proceeding with showing this, we first describe the relevant guarantees we can deduce from $n\geq N(q,t,\beta)$. First, 
\begin{equation}
\label{eq:g1}
\frac{1}{n}\leq \zeta.
\end{equation}
Second,
\begin{equation}
\label{eq:g2}
n\Delta - \log_q(n+1) \geq t > t-1.
\end{equation}
This is obtained by using $\ln(n+1)\leq \sqrt{n}$ for $n>0$. Thus,
\[
\log_q(n+1) \leq \frac{\sqrt{n}}{\ln q}\leq \frac{n\Delta}{2},
\]
with the last inequality derived using $n\geq \frac{4}{\Delta^2(\ln q)^2}$. Since we also have $n\geq 2t/\Delta$, 
\[
n\Delta-\log_q(n+1) \geq \frac{n\Delta}{2} \geq t.
\]

Let $\cC$ be a code, and write $s=\frac{n-k}{n}\in [\beta,1]$. We will now prove that
\begin{equation}
\label{eq:toprove}
\frac{d_t(\cC)}{2n} \leq H_{q^t}^{-1}(s)-\frac{\gamma}{2}.
\end{equation}
As in the proof of Lemma~\ref{lem:gap}, we write
\[
H_q(x_0)=(t+1)x_0 \eqdef s_0.
\]
We distinguish between two cases.

\textbf{Case 1:} $s \geq s_0$. We have
\begin{align*}
\frac{d_t(\cC)}{2n} &\overset{(a)}{\leq} \frac{s}{2}+\frac{t}{2n}
\overset{(b)}{\leq} \frac{\xi_{q,t}(s)}{2} + \frac{t\zeta}{2}
\overset{(c)}{\leq} \frac{\xi_{q,t}(s)}{2} + \frac{t\gamma}{4(t+1)}
< \frac{\xi_{q,t}(s)}{2} + \frac{\gamma}{4} \\
&\overset{(d)}{\leq} H_{q^t}^{-1}(s)-\frac{3\gamma}{4},
\end{align*}
where $(a)$ follows from Lemma~\ref{lem:knownsingleton}, $(b)$ follows from the definition of $\xi_{q,t}(s)$ and~\eqref{eq:g1}, $(c)$ follows from the definition of $\zeta$, and $(d)$ follows from the definition of $\gamma$. This chain implies~\eqref{eq:toprove}.

\textbf{Case 2:} $s<s_0$. Set $x=H_q^{-1}(s) < x_0$, as well as $r=\ceil{n(x+\zeta)}$. From~\eqref{eq:g1},
\begin{equation}
\label{eq:chain}
x+\zeta \leq \frac{r}{n} \leq x+\zeta+\frac{1}{n} \leq x_0+2\zeta \leq 1-\frac{1}{q}.
\end{equation}
In particular, $\frac{r}{n}\in (0,1-q^{-1})$. The $q$-ary entropy function is concave, and so $H_q(v+\zeta)-H_q(v)$ is decreasing in $v$. Since $x\leq x_0$, we obtain
\[
H_q(x+\zeta)-H_q(x) \geq H_q(x_0+\zeta)-H_q(x_0) = \Delta.
\]
Hence,
\[
H_q(r/n) \geq H_q(x+\zeta) \geq s+\Delta.
\]
Using a standard bound on the size of Hamming balls (e.g., \cite[Lemma 4.8]{Rot06}),
\begin{align*}
\cV_q(n,r) &\geq \frac{1}{n+1}q^{n H_q(r/n)}
\geq \frac{1}{n+1}q^{n-k+n\Delta} = q^{n-k+n\Delta-\log_q(n+1)} \\
&\geq q^{n-k+t} > q^{n-k+t-1},
\end{align*}
where the penultimate inequality follows from~\eqref{eq:g2}. Then
\begin{align*}
\frac{d_t(\cC)}{2n} &\overset{(a)}{\leq} \frac{(t+1)r}{2n}
\overset{(b)}{\leq} \frac{t+1}{2}\parenv*{x+\zeta+\frac{1}{n}}
\leq \frac{(t+1)x}{2} + (t+ 1)\zeta \\
& \overset{(c)}{\leq} \frac{\xi_{q,t}(s)}{2}+\frac{\gamma}{2}
\overset{(d)}{\leq} H_{q^t}^{-1}(s)-\frac{\gamma}{2}
\end{align*}
where $(a)$ follows by Lemma~\ref{lem:affball}, $(b)$ follows by~\eqref{eq:chain}, $(c)$ follows since by the definition of $\zeta$ we have $(t+1)\zeta\leq \gamma/2$, and $(d)$ follows from the definition of $\gamma$. This shows the proof of~\eqref{eq:toprove} in this case.

Recalling the standard ball-covering bound of~\eqref{eq:downball}, and the standard bound on the ball size using the entropy function (e.g., \cite[Lemma 4.7]{Rot06}),
\[
q^{tn H_{q^t}(R_t(\cC)/n)} \geq \cV_{q^t}(n,R_t(\cC)) \geq q^{t(n-k)}.
\]
If $R_t(\cC)/n > 1-q^{-t}$ then automatically $R_t(\cC)/n > H_{q^t}^{-1}(s)$. Otherwise, taking $\log_{q^t}$ of both sides and rearranging, we get
\[
    \frac{R_t(\cC)}{n} \geq H_{q^t}^{-1}\parenv*{\frac{n-k}{n}} = H_{q^t}^{-1}(s).
\]
Hence, with~\eqref{eq:toprove},
\[
R_t(\cC) - \frac{d_t(\cC)}{2} \geq \frac{\gamma n}{2} > 0.
\]
This proves that
\[
\delta_t(\cC) = \floor*{\frac{d_t(\cC)-1}{2}} \leq \frac{d_t(\cC)}{2} < R_t(\cC).
\]
\end{proof}

Armed with all the previous lemmas, we can now state and prove the main asymptotic theorem.

\begin{theorem}
\label{th:asymptotic}
For any real $0 < \alpha < 1$, there exists and integer $N_\alpha$, such that for every prime power $q$, every $n\geq N_\alpha$, and every $[n,k]_q$ code $\cC$ with rate $\frac{k}{n}\leq \alpha$, for all $1\leq t\leq \min\set{k,n-k}$,
\[
\delta_t(\cC) \leq R_t(\cC).
\]
\end{theorem}

\begin{proof}
Write $\beta = 1-\alpha \in (0,1)$. For $t=1$ the claim is classic, and for $t=2$ we have Theorem~\ref{th:t2}. By Corollary~\ref{cor:log}, for $t\geq 3$ all pairs $(q,t)$ are solved when $t\beta \log_2 q \geq 3$. The remaining unresolved pairs form the set
\[
\cP_\beta \eqdef \set*{ (q,t) : \text{$q$ is a prime power}, t\geq 3, t\beta \log_2 q < 3}.
\]
This set if obviously finite. For each pair $(q,t)\in\cP_\beta$ apply Lemma~\ref{lem:uniform}. Let
\[
N_\alpha = \max\parenv*{\set{1} \cup \set*{N(q,t,\beta) : (q,t)\in \cP_\beta}}.
\]
This choice guarantees all pairs in $\cP_\beta$ are handled when $n\geq N_\alpha$.
\end{proof}

\section{General Bounds}
\label{sec:genbound}

In this last section we bring two extensions of the method used for proving the conjecture holds when $t=2$ from Section~\ref{sec:t2}.

Let us fix an $\F_q$-linear isomorphism
$\psi:\F_q^t\to\F_{q^t}$. Thus, for each $\alpha\in\F_{q^t}$ we denote by
$M_\alpha\in\F_q^{t\times t}$ the matrix corresponding to
multiplication by $\alpha$, i.e., for all $x\in\F_q^t$,
\[ \psi(M_\alpha x)=\alpha \psi(x).\]

Consider an $[n,k]_q$ code $\cC$ and let $G\in\F_q^{k\times n}$ be a
generator matrix for $\cC$, i.e., $\cC=\rowsp G$. The point multiset
of $G$ is defined as the multiset of projective points appearing as
columns of $G$, namely,
\[
P_G \eqdef \mset{ \ang{G_{|j}}\in\PG(k-1,q) : G_{|j}\neq 0, 1\leq j\leq n}.
\]

\begin{theorem}
  \label{th:RtLambda}
  Let $\cC$ be an $[n,k]_q$ code, with generator
  matrix $G\in\F_q^{k\times n}$. Let $1\leq
  t\leq k$. Then
  \[ R_t(\cC)\geq \max_{\substack{U\in\F_q^{t\times k}\\ \rank U=t}}
  \sum_{\substack{\ang{v}\in\PG(k-1,q) \\ Uv\neq 0}}
  \parenv*{\#_{\ang{v}}(P_G)-\ceil*{\frac{\#_{\ang{v}}(P_G)}{q^t}}},\]
  and the RHS does not depend on the choice of generator matrix $G$.
\end{theorem}
\begin{proof}
  Fix some generator matrix $G\in\F_q^{k\times n}$ for $\cC$. All the
  codewords of $\cC^t$ are given by $UG$, where $U\in\F_q^{t\times
    k}$, and distinct choices of $U$ lead to distinct codewords. Fix
  some $U\in\F_q^{t\times k}$ such that $\rank U = t$, and let
  $C=UG\in\cC^t$. For any $\ang{v}\in P_G$ such that $Uv\neq 0$,
  partition the $\#_{\ang{v}}(P_G)$ coordinates of $G$ containing a
  column in the equivalence class $\ang{v}$, into $q^t$ sets,
  $A_{\ang{v},\alpha}$, for all $\alpha\in\F_{q^t}$, whose sizes differ by
  at most $1$. Thus, for all $\alpha\in\F_{q^t}$,
  \[ \abs*{A_{\ang{v},\alpha}} \leq \ceil*{\frac{\#_{\ang{v}}(P_G)}{q^t}}.\]
  Define $X\in\F_q^{t\times n}$ such that
  \[ X_{|j} = \begin{cases}
    M_\alpha C_{|j}, & j\in A_{\ang{v},\alpha}, \alpha\in\F_{q^t} \\
    0, & \text{otherwise.}
  \end{cases}\]

  Consider an arbitrary codeword in $\cC^t$, $Y=U'G\in\cC^t$, where
  $U'\in\F_q^{t\times k}$. Write $Z\eqdef C-Y=(U-U')G$. For a
  projective point $\ang{v}\in P_G$ such that $Uv\neq 0$, assume $j$
  is a coordinate for which $G_{j}= \gamma v$,
  $\gamma\in\F_q\setminus\set{0}$. Then
  \begin{align*}
    C_{|j}&=UG_{|j} = \gamma U v \neq 0, \\
    Y_{|j}&=U'G_{|j} = \gamma U' v.
  \end{align*}
  Set $\beta = \psi(\gamma U' v)/ \psi(\gamma U v) \in \F_{q^t}$. Then
  \[
  (X-Y)_{|j} = \begin{cases}
    0, & j\in A_{\ang{v},\beta}, \\
    \gamma M_\alpha U v - \gamma U' v\neq 0 & j\in A_{\ang{v},\alpha}, \alpha\neq \beta.
  \end{cases}
  \]
  Therefore, the equivalence class $\ang{G_{|j}}$ contributes at most
  $\ceil{\#_{\ang{G_{|j}}}(P_G)/q^t}$ zero columns in $X-Y$. Thus,
  \[
  \wt(X-Y) \geq \sum_{\substack{\ang{v}\in \PG(k-1,q) \\ Uv\neq 0}} \parenv*{\#_{\ang{v}}(P_G)-\ceil*{\frac{\#_{\ang{v}}(P_G)}{q^t}}}.
  \]
  Since this holds for any $Y\in\cC^t$, the distance of $X$ to the
  nearest codeword in $\cC^t$ is at least this expression. Maximizing
  over all possible $U$ proves the claim.

  Regarding the choice of $G$, if we replace $G$ with another
  generator $AG$, where $A\in\F_q^{k\times k}$ is invertible, we note
  that this simply permutes the equivalence classes that are the
  points of $\PG(k-1,q)$, and the bound remains the same.
\end{proof}

\begin{corollary}
  Let $\cC$ be an $[n,k]_q$ code. Then for every $1\leq t\leq k$,
  \[ R_t(\cC) \geq d_t(\cC)-\min_{\substack{U\in\F_q^{t\times k}\\ \rank U=t}}\sum_{\substack{\ang{v}\in \PG(k-1,q) \\ Uv\neq 0}} \ceil*{\frac{\#_{\ang{v}}(P_G)}{q^t}}.
  \]
\end{corollary}
\begin{proof}
Proceed as in the proof of Theorem~\ref{th:RtLambda}, but at the end, recall that for all
$U\in\F_q^{t\times k}$ for which $\rank U=t$, we have $\wt(UG)\geq d_t(\cC)$.
\end{proof}

\begin{corollary}
Let $\cC$ be an $[n,k]_q$ code. If $\#_{\ang{v}}(P_G)$ is a multiple of $q^t$ for all $v$, then
\[
R_t(\cC)\geq \parenv*{1-\frac{1}{q^t}}d_t(\cC).
\]
\end{corollary}

A nicer form of result, similar in flavor to the previous theorem, uses the probabilistic method.

\begin{theorem}
  Let $\cC$ be an $[n,k]_q$ code, and let $1\leq t\leq k$. Then
  \[
  R_t(\cC) \geq \ceil*{\parenv*{1-\frac{1}{q^t}}d_t(\cC)-\sqrt{\frac{d_t(\cC)}{2}tk\ln q}}.
  \]
\end{theorem}
\begin{proof}
  Choose an arbitrary set $S\subseteq\set{1,\dots,n}$ of coordinates
  of size $\abs{S}=d_t(\cC)$. Construct a random matrix
  $X\in\F_q^{t\times n}$ by choosing each column $X_{|j}$
  independently and uniformly from $\F_q^t$, and then setting
  $X_{|j}=0$ for all $j\notin S$.

  For any $j\in S$, define the indicator random variable $I_j$ to be
  $1$ if and only if $X_{|j}\neq Y_{|j}$. Set $Z_Y \eqdef \sum_{j\in
    S} I_j$. Thus, $Z_Y$ measures the contributions of the coordinates
  in $S$ to the distance between $X$ and $Y$, and so $\wt(X-Y)\geq
  Z_Y$. Clearly, $Z_Y$ has a binomial distribution, $Z_Y \sim
  \Bin(d_t(\cC),1-q^{-t})$, since it is the sum of i.i.d. Bernoulli
  random variables with probability $1-q^{-t}$. It follows that
  \[ \E [Z_Y] = \parenv*{1-\frac{1}{q^t}} d_t(\cC).\]
  By definition of the generalized covering radius, every instance of
  $X$ is within distance $R_t(\cC)$ from at least one $Y\in\cC^t$. Hence,
  by the union bound
  \begin{align}
    1 &= \Pr\parenv*{\bigcup_{Y\in\cC^t} \set*{\wt(X-Y)\leq R_t(\cC)}}
    \leq \sum_{Y\in\cC^t} \Pr( Z_Y\leq R_t(\cC)). \label{eq:ZY}
  \end{align}
  If $R_t(\cC)\geq \E[Z_Y]=(1-q^{-t})d_t(\cC)$, then we already have
  the claim proved. Otherwise, by Hoeffding's inequality, using $Z_Y$
  as a sum of i.i.d Bernoulli random variables,
  \begin{align*}
    \Pr(Z_Y\leq R_t(\cC)) & = \Pr( Z_Y-\E[Z_Y] \leq -(\E[Z_Y]-R_t(\cC))) \\
    & \leq
    \exp\parenv*{-\frac{2(\E[Z_Y]-R_t(\cC))^2}{d_t(\cC)}}.
  \end{align*}
  With~\eqref{eq:ZY} we have
  \[ 1 \leq q^{kt}\exp\parenv*{-\frac{2(\E[Z_Y]-R_t(\cC))^2}{d_t(\cC)}}.\]
  Using $\E[Z_Y]=(1-q^{-t})d_t(\cC)$ and rearranging, gives us the
  desired claim.
\end{proof}

\section*{Acknowledgments}
This research was support in part by the Natural Sciences and Engineering Research Council of Canada (NSERC) under grant no.~RGPIN-2026-06805.

\bibliographystyle{elsarticle-num}
\bibliography{allbib}

\end{document}